\documentclass[11pt,runningheads,a4paper]{llncs}

\usepackage[english]{babel}
\usepackage[T1]{fontenc}
\usepackage{lmodern}
\usepackage{amssymb}

\newcommand\vnp{\ensuremath{\mathsf{VNP}}}

\newcommand\zz{\ensuremath{\mathbb{Z}}}
\newcommand\nn{\ensuremath{\mathbb{N}}}

\newcommand\p{\ensuremath{\mathsf{P}}}

\newcommand\per{\ensuremath{\mathrm{PER}}}
\newcommand\sharpp{\ensuremath{\mathsf{\sharp P}}}
\newcommand\pp{\ensuremath{\mathsf{PP}}}
\newcommand\rp{\ensuremath{\mathsf{RP}}}

\newcommand\nexp{\ensuremath{\mathsf{NEXP}}}

\newcommand\e{\ensuremath{\mathsf{E}}}
\newcommand\aczero{\ensuremath{\mathsf{AC}^0}}
\newcommand\tczero{\ensuremath{\mathsf{TC}^0}}

\newcommand\dlogtime{\ensuremath{\mathsf{DLOGTIME}}}
\newcommand\dtime{\ensuremath{\mathsf{DTIME}}}

\begin{document}

\mainmatter

\title{A Superpolynomial Lower Bound on the Size of Uniform Non-constant-depth
  Threshold Circuits for the Permanent} \titlerunning{A Lower Bound for the
  Permanent}

\author{Pascal Koiran\inst{1}\and Sylvain Perifel\inst{2}}

\institute{LIP, École Normale Supérieure de Lyon\\
{\tt Pascal.Koiran@ens-lyon.fr}
\and
LIAFA, Université Paris Diderot - Paris 7\\
{\tt Sylvain.Perifel@liafa.jussieu.fr}}


\date{\today}

\maketitle

\begin{abstract}
  We show that the permanent cannot be computed by \dlogtime-uniform threshold
  or arithmetic circuits of depth $o(\log \log n)$ and polynomial size.\bigskip

  \textbf{Keywords.} Permanent, lower bound, threshold circuits, uniform
  circuits, non-constant depth circuits, arithmetic circuits.
\end{abstract}

\section{Introduction}

Both in Boolean and algebraic complexity, the permanent has proven to be a
central problem and showing lower bounds on its complexity has become a major
challenge. This central position certainly comes, among others, from its
\sharpp-completeness~\cite{valiant1979b}, its
\vnp-completeness~\cite{valiant1979}, and from Toda's theorem stating that the
permanent is as powerful as the whole polynomial
hierarchy~\cite{toda1989}. More recently, it played a role in the celebrated
and subtle result of Kabanets and Impagliazzo~\cite{KI2004}: either
$\nexp^\rp$ does not have Boolean circuits of polynomial size, or the
permanent does not have arithmetic circuits of polynomial size.

However little is known on the circuit complexity of the permanent in the
general case. Indeed, the best lower bound so far on its circuit size is no
more than the trivial $\Omega(n^2)$ (remember that $\per_n$ has $n^2$
variables). Despite this rather dark state of affairs, some progress has been
made on restricted classes of circuits. For instance, we know lower bounds on
\emph{monotone} circuits (such circuits for the permanent must have
exponential size, see~\cite{JS1982,SV1998}), and recently, lower bounds on
\emph{multilinear} circuits were obtained (see
e.g.~\cite{raz2004,RSY2007,RY2008}).

A lot of work has also been done on \emph{constant-depth} circuits, in which
gates have unbounded fan-in. This line of research has been quite successful
on Boolean circuits and gave deep insights into circuit complexity: see
e.g.~\cite{FSS1984,smolensky1987}. However, pushing the limit of lower bounds
beyond constant depth for polynomial-size circuits has remained elusive
so far.

Another restriction worth studying is \emph{uniformity}: circuits are not
arbitrary any more but are required to be described by a Turing machine. If
this description is very efficient (running in time logartihmic in the size of
the circuit, we speak of \dlogtime-uniformity), Allender~\cite{allender1999}
(see also similar results on circuits with modulo gates in~\cite{AG1994}) has
shown that the permanent does not have threshold circuits of \emph{constant}
depth and subexponential size. In this paper, we obtain a tradeoff between
size and depth: instead of subexponential size, we only prove a
superpolynomial lower bound on the size of the circuits, but now the depth is
no more constant. More precisely, we show the following theorem.
\begin{theorem}\label{th_main}
  The permanent does not have \dlogtime-uniform polynomial-size threshold
  circuits of depth $o(\log\log n)$.
\end{theorem}
It seems to be the first superpolynomial lower bound on the size of
\textit{non-constant-depth} threshold circuits for the permanent (though a
lower bound is proved in~\cite{RY2008} on \emph{multilinear} arithmetic
circuits of depth $o(\log n/\log \log n)$). Admittedly, the depth $o(\log\log
n)$ is still small but until now the known techniques were only able to prove
lower bounds on constant-depth circuits.

Let us very briefly describe our proof technique. In contrast
with~\cite{allender1999}, we do not use the relation between threshold
circuits and the counting hierarchy, which implied to consider only
constant-depth circuits. Also, the diagonalization in~\cite{allender1999} is a
variant on the nondeterministic time hierarchy theorem. Here, we use the usual
deterministic time hierarchy theorem as an indirect diagonalization : under
the assumption that the permanent has \dlogtime-uniform circuits of polynomial
size and depth $o(\log \log n)$, we show
\begin{enumerate}
\item the value of a threshold circuit of size $s$ and depth $d$ can be
  computed in time $(\log s)^{2^{O(d)}}$ (Lemma~\ref{lemma_main} combined with
  Lemma~\ref{lemma_per});
\item every language in \e\ has uniform threshold circuits of size
  $2^{O(n)}$ and depth $o(\log n)$ (Corollary~\ref{cor_thresh_E}).
\end{enumerate}
These two points together imply that every language in \e\ can be
computed in subexponential time, a contradiction with the time hierarchy
theorem.

Since threshold circuits can simulate arithmetic circuits, we also obtain a
superpolynomial lower bound on the size of uniform arithmetic circuits of
depth $o(\log\log n)$ for the permanent (Corollary~\ref{cor_arith}).

\textit{Organization of the paper --- } The next section is devoted to the
definition of the notions in use: circuits (Boolean, threshold, arithmetic),
uniformity and some complexity classes. Then Section~\ref{sec_technical} is
dedicated to the proof of Theorem~\ref{th_main} by showing a series of results
along the way suggested above.

\section{Preliminaries}

The notions we use are very standard but, for completeness, we still recall
them in this section.

\subsection{Boolean circuits}

A Boolean circuit on $n$ variables is a directed acyclic graph, whose vertices
are labeled either by a variable among $\{x_1,\dots,x_n\}$ or by an operation
among $\{\lor,\land,\lnot\}$. Vertices of indegree\footnote{Indegree and
  outdegree are also called fan-in and fan-out, respectively.} 0 are called
inputs, the others are called gates. A gate labeled by $\lnot$ is required to
have indegree 1, whereas gates labeled by $\lor$ or $\land$ have indegree 2. A
single gate has outdegree 0 and is called the output gate.

The value computed by a vertex is defined recursively: an input $x_i$ has for
value the value of the variable $x_i\in\{0,1\}$. A $\lnot$ gate $g=\lnot h$
has for value the negation of the value of $h$. An $\lor$ gate $g=h_1\lor h_2$
(respectively an $\land$ gate $g=h_1\land h_2$) has for value the disjunction
(resp. conjunction) of the values of $h_1$ and $h_2$. The value of the circuit
is by definition the value of its output gate.

The size of the circuit is the number of vertices and the depth is the length
of the longest path from an input vertex to the output gate.

Remark that in order to recognize a language, one needs not only one but a
whole family (that is, an infinite sequence) of circuits $(C_n)$, as explained
below. There is also a variant in which gates $\lor$ and $\land$ have
unbounded fan-in: this is useful when defining classes of circuits of
constant depth.

\subsection{Threshold circuits}

A threshold circuit has a similar definition as a Boolean circuit with $\lor$
and $\land$ gates of arbitrary fan-in, but another type of gates is allowed:
threshold gates (also known as majority gates). A threshold gate is also of
arbitrary fan-in, and its value is 1 if at least half of its inputs have value
1, and 0 otherwise.

Again, in order to recognize a language, a whole family of circuits is
needed. Remark that it makes sense to consider families of bounded depth
threshold circuits since gates are allowed to have arbitrary fan-in.

\subsection{Arithmetic circuits}

An arithmetic circuit is defined similarly as a Boolean circuit but with other
kinds of gates. It has $+$, $-$ and $\times$ gates, all of fan-in 2, and
besides variables, another input is labeled by the constant 1. The variables
are not considered to have Boolean values anymore, but instead they are
symbolic and the circuit computes a polynomial (over the ring $\zz$) in the
obvious way: the value of the input gate labeled by 1 is the constant
polynomial 1, the value of an input gate labeled by $x_i$ is the polynomial
$x_i$, the value of a $+$ gate (respectively $-$ gate, $\times$ gate) is the
sum (resp. difference, product) of the values of its inputs.

An arithmetic circuit $C$ with $n$ input gates computes a multivariate
polynomial over \zz\ with $n$ variables. Circuit families $(C_n)$ are used to
compute families of polynomial. The permanent family (also called permanent
for short) is the family $(\per_n)$ of polynomials defined as follows:
$$\per_n(x_{1,1},x_{1,2},\dots,x_{1,n},x_{2,1},\dots,x_{n,n})=
\sum_{\sigma}\prod_{i=1}^n x_{i,\sigma(i)}$$ where the sum is taken over all
the permutations $\sigma$ of $\{1,\dots,n\}$. The $n^2$ variables $x_{i,j}$
can be viewed as the coefficients of an $n\times n$ matrix, allowing us to
speak of the permanent of a matrix.

\subsection{Uniformity}

Circuits, be they Boolean, threshold or arithmetic, are finite objects easily
encoded in binary (e.g. by the list of their vertices and edges). Hence they
can be handled by Turing machines.

As already mentioned, we are interested in sequences $(C_n)$ of circuits in
order to recognize languages. In whole generality, no assumption is made on
the structure of these circuits: in particular, the Boolean encodings of the
circuits of a family may be uncomputable. However, if a single Turing machine
is able to produce the Boolean encoding of all the circuits of the family,
then we speak of uniformity. The degree of uniformity depends on the
ressources needed by the machine.

A family of circuits $(C_n)$ is said \p-uniform if there exists a
deterministic Turing machine which, on input $(n,i)$ given in binary, outputs
the $i$-th bit of the encoding of $C_n$ in time polynomial in $n$ (that is, in
time exponential in the size of the input). Similarly, a family of circuits
$(C_n)$ is said \dlogtime-uniform if there exists a deterministic Turing
machine which, on input $(n,i)$ given in binary, outputs the $i$-th bit of the
encoding of $C_n$ in time logarithmic in $n$ (that is, in time linear in the
size of the input). Of course, \dlogtime-uniformity implies \p-uniformity. It
can be argued that \dlogtime-uniformity is the right notion of uniformity for
small-depth circuits, see~\cite{BIS1990}.

\begin{remark}
  In the remainder of the paper, we shall work with \dlogtime-uniformity, but
  everything remains valid if replaced by ``polylogtime'' uniformity.
\end{remark}

\subsection{Complexity classes}

Finally, we will meet some complexity classes defined now. Let $\dtime(t(n))$
denote the set of languages recognized in time $t(n)$ by a deterministic
Turing machine. Then \p\ is the class $\dtime(n^{O(1)})=\cup_{k>0}\dtime(n^k)$
(that is, deterministic polynomial time) and \e\ is the class
$\dtime(2^{O(n)})=\cup_{k>0}\dtime(2^{k.n})$ (that is, deterministic
exponential time with linear exponent).

Recall the time hierarchy theorem~\cite{HS1965}: for time-constructible
functions $f$ and $g$, if $f(n)/g(n)=o(1/\log(g(n)))$ then
$\dtime(g(n))\not\subset\dtime(f(n))$. In particular, we will use the
following consequence: $\e\not\subset\dtime(n^{2^{o(\log n)}})$.

The class \sharpp\ is the set of functions $f:\{0,1\}^*\to\nn$ defined as
follows: there exist a polynomial $p$ and a language $A\in\p$ such that
$f(x)=\#\{y\in\{0,1\}^{p(|x|)}:(x,y)\in A\}$. Computing the permanent of a 0-1
matrix is \sharpp-complete (Valiant~\cite{valiant1979}). Then \pp\ is the set
of languages $B$ such that there is $f\in\sharpp$ satisfying $[x\in B\iff
f(x)\geq 2^{p(|x|)-1}]$. The class \pp\ can also be viewed as the languages
$B$ such that there exists a polynomial-time nondeterministic Turing machine
$N$ satisfying [$x\in B$ iff at least half of the computation paths of $N$ are
accepting]. Remark that if every function in \sharpp\ can be computed in
polynomial time, then $\pp=\p$.

Complexity classes can also be defined in terms of circuits (either Boolean or
threshold). An input $x$ is accepted by a circuit $C$ if the value of $C$ on
$x$, denoted by $C(x)$, is 1. In order to recognize languages, families
$(C_n)$ of circuits are considered: circuit $C_n$ will recognize inputs of
size $n$, hence we make the assumption that $C_n$ has $n$ input gates. Now, a
language $A$ is recognized by a family $(C_n)$ of circuits if
$A=\{x\in\{0,1\}^*:C_{|x|}(x)=1\}$.

We shall use the well known characterization of \p\ in terms of circuits: \p\
is the set of languages recognized by \p-uniform families of polynomial-size
Boolean circuits. The class \aczero\ is the set of languages recognized by a
family of constant-depth Boolean circuits of polynomial size, where the gates
$\lor$ and $\land$ have unbounded fan-in. The class \tczero\ is the set of
languages recognized by a family of constant-depth threshold circuits of
polynomial size. Uniform versions of these classes, \dlogtime-\aczero\ and
\dlogtime-\tczero\ respectively, are defined by requiring \dlogtime-uniformity
on the circuit family.

\section{Technical developments}\label{sec_technical}

\noindent This series of results is devoted to the proof of
Theorem~\ref{th_main}. 


\begin{lemma}\label{lemma_per}
  If the permanent has \p-uniform polynomial-size threshold circuits
  then \pp=\p.
\end{lemma}
\begin{proof}
  First turn the threshold circuits into Boolean circuits. To this end, every
  $\land$ or $\lor$ gate of unbounded fan-in is replaced by trees of $\land$
  or $\lor$ gates of fan-in 2 (which clearly remains \p-uniform and of
  polynomial size), and every threshold gate with $N=n^{O(1)}$ inputs is
  replaced by the addition of the inputs followed by a comparison of the
  result with $N/2$. This iterative addition can easily be carried out by a
  \p-uniform circuit of size polynomial in $N$, hence polynomial in $n$. This
  proves that the permanent has \p-uniform polynomial-size Boolean circuits.

  Thus, by \sharpp-completeness of the permanent every function in \sharpp\
  can be computed in polynomial time. This implies that $\pp=\p$.\qed
\end{proof}
As a preparation to the proof of Lemma~\ref{lemma_main}, let us first rephrase
the hypothesis $\pp=\p$ in a convenient way.
\begin{lemma}\label{lemma_pp}
  Let $A$ be a language with a (deterministic) algorithm running in time
  $t(n)\geq n$. Consider the following problem $B$: given a word $x$, a length
  $n$ and an integer $N\leq 2^n$, decide
  whether at least $N$ words $y$ of size $n$ satifsy $(x,y)\in A$.\\
  If $\pp=\p$ then $B$ has an algorithm running in time $p(t(n))$ for a fixed
  polynomial $p$ (independent of $A$).
\end{lemma}
\begin{proof}
  Remark that this is not a completely obvious consequence of $\pp=\p$ since
  the polynomial $p$ is required to be independent of $A$. In fact this comes
  from the existence of a complete problem for \pp. Take indeed the canonical
  \pp-complete language $H=\{(M,x,1^n):\mbox{at least half of the computation
    paths of $M(x)$ are accepting in time $n$}\}$, where $M$ is a
  nondeterministic Turing machine. The hypothesis $\pp=\p$ implies that $H$ is
  decidable in time $p(n)$ for some polynomial $p$.

  To the problem $B$ is associated a language
  $\tilde{B}=\{(x,n,N,1^{t(n)}):\#\{y\in\{0,1\}^n:(x,y)\in A\}\geq N\}$. Then
  $\tilde{B}$ is in \pp\ and a reduction from $\tilde{B}$ to $H$ is the
  mapping $(x,n,N,1^{t(n)})\mapsto (M,(n,N,x),1^{t(n)})$, where $M(n,N,x)$ has
  the following behaviour: it guesses a bit $b\in\{0,1\}$; if $b=0$ then it
  creates $2^n-N$ accepting paths among $2^n$ paths; if $b=1$ then it guesses
  $y\in\{0,1\}^n$ and decides whether $(x,y)\in A$ by running the algorithm
  for $A$ in time $t(n)$. Therefore $M(n,N,x)$ runs in time $O(t(n))$, has
  $2^{n+1}$ paths, and among them $\#\{y\in\{0,1\}^n:(x,y)\in A\}+(2^n-N)$ are
  accepting. This is at least half iff $\#\{y\in\{0,1\}^n:(x,y)\in A\}\geq
  N$. This reduction shows that $\tilde{B}$ is decidable in time
  $p(t(n))$.\qed
\end{proof}
Similarly as succinct representations used for exponential-time-complete
languages, threshold circuits can be succinctly given, not by their binary
encoding but rather by a description of their gates. That is, instead of
giving the threshold circuit $C$ directly, a Boolean circuit $B$ is given,
whose value $B(i)$ on input $i$ is the $i$-th bit of the encoding of $C$. This
may enable to give a much shorter representation of the circuit. Circuits
given in that way will be called ``succinctly given''.
\begin{lemma}\label{lemma_main}
  Let $A$ be the problem of deciding the value of a succinctly given
  threshold circuit, that is,
  $$A=\{(B,x): B\mbox{ represents a threshold circuit $C$ and $C(x)=1$}\}$$
  where $B$ is a Boolean circuit and $x$ is a Boolean input to $C$ of
  appropriate size. The size of the threshold circuit $C$ is
  denoted by $s$ and its depth by $d$. Suppose furthermore that the size of
  the input $(B,x)$ is less than $(\log s)^{2^d}$.\\
  If $\pp=\p$, then $A$ has an algorithm of running time $(\log
  s)^{2^{O(d)}}$.
\end{lemma}
\begin{proof}
  The idea is to recursively evaluate the values of the gates at each depth of
  the circuit, using Lemma~\ref{lemma_pp} for threshold gates. In order to
  apply Lemma~\ref{lemma_pp}, one has to consider all the inputs of a
  particular gate, leading us to define the language $A_k$ corresponding to
  the gates being inputs of the $i$-th gate of $C$, whose depth is $\leq k$,
  as follows:
  $$\begin{array}{ll}
    A_k=\{(B,x,i,j): & B\mbox{ represents a threshold circuit $C$ in which}\\
    & \mbox{gate number $i$ is at depth $\leq k$,}\\
    & \mbox{gate number $j$ is an input of gate $i$, and}\\
    & \mbox{the value of gate $j$ in the computation $C(x)$
      is 1}\}
  \end{array}$$
  Note that one can artificially add to $C$ a final ``identity gate'' taking as
  input the output of $C$, in which case deciding $A_{d+1}$ implies computing
  the value of $C(x)$.

  Let us call $T(k,d,s)$ the time needed to decide $A_k$ as a function of the
  size $s$ and the depth $d$ of $C$. The language $A_2$ merely consists in
  evaluating an input gate, that is, deciding to which bit of $x$ it
  corresponds: this can be done in polynomial time, hence in time $(\log
  s)^{O(2^d)}$ by assumption on the size of $(B,x)$. Therefore
  $T(2,d,s)=(\log s)^{O(2^d)}$.


  The purpose is now to decide $A_{k+1}$ by using the algorithm for $A_k$. It
  can be done easily since we can decide the value of a gate at depth $k$ if
  we know the values of the gates at depth $\leq k-1$. Indeed, let us decide
  whether $(B,x,i,j)\in A_{k+1}$, supposing gate $i$ is at depth $k+1$ and has
  gate $j$ as input: we want to compute the value of gate $j$. Since gate $j$
  is at depth $\leq k$, the algorithm for $A_k$ provides the value of all the
  inputs of gate $j$, which are used in turn to compute the value of gate $j$
  itself. More precisely we proceed inductively:
  \begin{itemize}
  \item If gate $j$ is a $\lnot$ gate, that is, $f=\lnot g$, then the value of
    $f$ is the negation of the value of $g$.
  \item If gate $j$ is an $\lor$ or an $\land$ gate, that is, $f=g\diamond h$
    with $\diamond\in\{\lor,\land\}$, then we perform the corresponding
    Boolean operation on the values of $g$ and $h$.
  \item Finally, if gate $j$ is a threshold gate, it has at most $s-1$ inputs
    and we decide whether at least half of them evaluate to 1.
  \end{itemize}
  Let us bound the execution time $T(k+1,s,d)$ of this algorithm for $A_{k+1}$
  as a function of $T(k,s,d)$ (the execution time of the algorithm for
  $A_k$). In the first case, we take the negation of one request of the form
  $(B,x,j,g)\in A_k$, therefore we have the following relation:
  $T(k+1,s,d)=T(k,s,d)+O(1)$. Similarly, in the second case we make a Boolean
  combination of two requests (one for each input), hence $T(k+1,s,d)\leq
  2T(k,s,d)+O(1)$. Finally in the third case, the task is to decide whether
  more than half of the inputs $y$ of gate $j$ evaluate to 1. Applying
  Lemma~\ref{lemma_pp} to the language $A_k$ with requests of the form
  $(B,x,j,y)\in A_k$ for all gates $y$ input of $j$, yields $T(k+1,s,d)\leq
  p(T(k,s,d))$ for some fixed polynomial $p$.

  As a whole, we have the following relation, for a fixed polynomial $p$:
  $$T(k+1,s,d)\leq p(T(k,s,d)).$$ In other words, there exists an exponent
  $\alpha\in\nn$ such that $T(k+1,s,d)\leq T(k,s,d)^{\alpha}$, hence
  $T(k,s,d)\leq T(2,s,d)^{\alpha^k}$. Since $T(2,s,d)=(\log s)^{O(2^d)}$ and
  deciding $A$ requires to go up to $k=d+1$, there is an algorithm for $A$
  running in time
  $$T(d+1,s,d)=(\log s)^{2^{O(d)}}.$$\qed
\end{proof}
Lemma~\ref{lemma_main} concerns the evaluation of succinctly given threshold
circuits; the consequence for languages is the following.
\begin{corollary}\label{cor_algo}
  Suppose a language $A$ has threshold circuits of size $s(n)$, depth
  $d(n)$ and constructible in polynomial time (that is, the $i$-th bit of
  $C_n$ is computable in time $n^{O(1)}$). Suppose furthermore that $(\log
  s(n))^{2^{O(d(n))}}$ is superpolynomial in $n$.\\
  If $\pp=\p$, then $A$ has an algorithm of running time $(\log
  s(n))^{2^{O(d(n))}}$.
\end{corollary}
Let us now see how to relate the hypothesis on the permanent to decision
languages. We need the following result concerning the completeness of the
permanent under a very strong notion of reduction. This result appears
in~\cite{AG1994} as a careful analysis of the usual reduction of
Valiant~\cite{valiant1979b} (see also~\cite{zanko1991} for many-one
reductions), which can in fact be carried out in a much more efficient way
than just polynomial time.
\begin{proposition}
  The permanent of 0-1 matrices is hard for \sharpp\ under \dlogtime-uniform
  \aczero\ many-one reductions, that is, the reduction is computed by
  \dlogtime-uniform \aczero\ circuits.
\end{proposition}
\begin{corollary}
  Every language $A\in\p$ can be expressed as the permanent of a 0-1 matrix
  $M$ of size $n^{O(1)}$, computed by \dlogtime-uniform \aczero\
  circuits. More precisely, there are functions $M$ and $\alpha$ computed by
  \dlogtime-uniform \aczero\ circuits such that $x\in
  A\Rightarrow\alpha(\per(M(x)))=1$ and $x\notin
  A\Rightarrow\alpha(\per(M(x)))=0$.
\end{corollary}
Scaling up this result to exponential time yields the following corollary.
\begin{corollary}\label{cor_matrix}
  For every language $A\in\e$, there are two functions $M$ and $\alpha$
  computable by size $2^{O(n)}$, constant-depth Boolean circuits constructible
  in polynomial time (that is, the $i$-th bit of the circuit is computable in
  time $n^{O(1)}$), such that $x\in A\Rightarrow\alpha(\per(M(x)))=1$ and
  $x\notin A\Rightarrow\alpha(\per(M(x)))=0$.
\end{corollary}
This implies the following result.
\begin{corollary}\label{cor_thresh_E}
  If the permanent has \dlogtime-uniform polynomial-size threshold circuits of
  depth $d(n)$, then every language $A$ in \e\ has threshold circuits of size
  $2^{O(n)}$ and depth $O(d(2^{O(n)}))$, these circuits being constructible in
  polynomial time (that is, the $i$-th bit of $C_n$ is computable in time
  $n^{O(1)}$).
\end{corollary}
\begin{proof}
  By Corollary~\ref{cor_matrix}, membership to $A$ is decided by the permanent
  of a matrix $M(x)$ of size $2^{O(n)}$. It is enough to compute the matrix
  $M(x)$ by constant-depth uniform circuits, then to plug the result into the
  uniform threshold circuits of depth $d(2^{O(n)})$ for the permanent of
  matrices of size $2^{O(n)}$, and finally to apply function $\alpha$ computed
  by constant-depth uniform circuits. The remaining circuits are again uniform
  threshold circuits of depth $O(d(2^{O(n)}))$.\qed
\end{proof}
Combining Corollary~\ref{cor_thresh_E}, Lemma~\ref{lemma_per} and
Corollary~\ref{cor_algo} yields the following.
\begin{corollary}
  If the permanent has \dlogtime-uniform polynomial-size threshold circuits of
  depth $d(n)$, then $\e\subseteq\dtime(n^{2^{O(d(2^{O(n)}))}})$.
\end{corollary}
This is in contradiction with the time hierarchy theorem as soon as
$d(n)=o(\log\log n)$, hence we have proved our main result:
\begin{corollary}[Theorem~\ref{th_main}]\label{cor_final}
  The permanent does not have \dlogtime-uniform polynomial-size threshold
  circuits of depth $o(\log \log n)$.
\end{corollary}
Since an arithmetic circuit can be simulated by a threshold one (addition and
multiplication are indeed in \dlogtime-uniform \tczero), we obtain the
following corollary.
\begin{corollary}\label{cor_arith}
  The permanent does not have \dlogtime-uniform polynomial-size arithmetic
  circuits of depth $o(\log \log n)$.
\end{corollary}\bigskip

\noindent\textit{Acknowledgments ---} The authors want to thank Eric Allender
for useful discussions.

\bibliographystyle{abbrv}
\bibliography{permanent}

\end{document}